\documentclass[letterpaper, 10 pt, conference]{ieeeconf}  

\IEEEoverridecommandlockouts                              
\usepackage{cite}
\usepackage{amsmath,amssymb,amsfonts,bbm}
\usepackage{algorithmic}
\usepackage{graphicx}
\usepackage{textcomp}
\usepackage{setspace}
\usepackage{booktabs}

\usepackage{epsfig}
\usepackage{times} 
\usepackage{svg}
\usepackage[linesnumbered,ruled,vlined]{algorithm2e} 

\usepackage[nolist, nohyperlinks]{acronym}
\usepackage[table]{xcolor} 
\usepackage{makecell}

\usepackage{amsthm}

\newcommand{\reals}{\mathbb{R}}
\newcommand{\R}{\reals}

\newcommand{\xbf}{\mathbf x}
\newcommand{\ubf}{\mathbf u}

\newcommand{\Qbf}{\mathbf Q}
\newcommand{\Vbf}{\mathbf V}

\newcommand{\Pbb}{\mathbb P}
\newcommand{\Ebb}{\mathbb E}
\newcommand{\Zbb}{\mathbb Z}

\newcommand{\Bcal}{\mathcal{B}}
\newcommand{\Ccal}{\mathcal{C}}

\newcommand{\Ecal}{\mathcal{E}}
\newcommand{\Fcal}{\mathcal{F}}
\newcommand{\Gcal}{\mathcal{G}}

\newcommand{\Kcal}{\mathcal{K}}
\newcommand{\Lcal}{\mathcal{L}}

\newcommand{\Ncal}{\mathcal{N}}

\newcommand{\Pcal}{\mathcal{P}}

\newcommand{\Scal}{\mathcal{S}}

\newcommand{\Ucal}{\mathcal{U}}
\newcommand{\Vcal}{\mathcal{V}}

\newcommand{\Xcal}{\mathcal{X}}

\newcommand{\eqn}[1]{\begin{align} #1 \end{align}}
\newcommand{\eqnN}[1]{\begin{align*} #1 \end{align*}}

\newcommand{\abs}[1]{\left | #1 \right |}

\theoremstyle{plain}
\newtheorem{theorem}{Theorem}

\newtheorem{lemma}{Lemma}
\newtheorem{problem}{Problem}
\newtheorem{definition}{Definition}
\newtheorem{prop}{Proposition}

\theoremstyle{definition}
\newtheorem{assumption}{Assumption}
\newtheorem{remark}{Remark}

\theoremstyle{remark}

\makeatletter
\let\NAT@parse\undefined
\makeatother
\usepackage{hyperref}
\hypersetup{
colorlinks, 
    linkcolor={red!50!black},
    citecolor={blue!50!black},
    urlcolor={blue!80!black}
}
\usepackage{cleveref}

\crefformat{problem}{Problem~#2#1#3}
\crefformat{assumption}{Assumption~#2#1#3}
\crefformat{prop}{Proposition~#2#1#3}
\usepackage[nolist,nohyperlinks]{acronym}

\title{\LARGE \bf
Fully Byzantine-Resilient Distributed Multi-Agent Q-Learning
}

\author{Haejoon Lee and Dimitra Panagou
\thanks{*This work is supported by the Air Force Office of Scientific Research (AFOSR) under FA9550-23-1-0163.}
\thanks{All authors are with the Department of Robotics,
        University of Michigan, Ann Arbor, MI, USA
        {\tt\small \{haejoonl, dpanagou\}@umich.edu}}%
}

\begin{document}

\maketitle
\thispagestyle{empty}
\pagestyle{empty}

\begin{abstract}
We study Byzantine-resilient distributed multi-agent reinforcement learning (MARL), where agents collaboratively learn optimal value functions over a compromised communication network. Existing resilient MARL approaches typically guarantee almost sure convergence only to near-optimal value functions, or require restrictive assumptions to ensure convergence to the optimal solution. Thus, agents may fail to learn the optimal policies under these methods. To address this, we propose a novel distributed Q-learning algorithm, under which all agents’ value functions converge almost surely to the optimal value functions despite Byzantine edge attacks. The key idea is a redundancy-based filtering mechanism that leverages two-hop neighbor information to validate incoming messages, while preserving bidirectional information flow. We then introduce a new topological condition for the convergence of our algorithm, present a systematic method to construct such networks, and prove that this condition can be verified in polynomial time. We validate our results through simulations, showing that our method converges to the optimal solutions, whereas other methods fail under Byzantine edge attacks.
\end{abstract}

\section{Introduction}
Multi-agent reinforcement learning (MARL) has emerged as a powerful extension of reinforcement learning (RL)~\cite{sutton1998reinforcement} for learning optimal policies of multiple decision-makers interacting within a shared environment~\cite{zhang2021multi}. Within this field, distributed cooperative MARL has gained significant attention. In these scenarios, multiple agents aim to optimize a shared global objective -- often expressed as the average of local rewards -- through information sharing with their neighbors in a network~\cite{zhang2018fully, lin2021multi, qu2020scalable}.

 One of the challenges in these decentralized frameworks is that agents receive local rewards from the environment, making the computation of optimal global value functions impossible without information sharing. To this end, a variety of frameworks have been developed. A consensus-based distributed Q-learning (QD-learning) has been proposed in~\cite{kar2013QD}, enabling agents to asymptotically compute the average of their optimal value functions through exchanges of state-action value estimates. Subsequent work introduced distributed actor-critic frameworks using linear function approximations~\cite{zhang2018fully,zeng2022learning}. To ensure scalability,~\cite{lin2021multi, qu2020scalable} studied scalable actor-critic algorithms where agents maintain state-action information only for their multi-hop neighbors.


Despite these merits, distributed MARL algorithms, just like other distributed algorithms, are highly vulnerable to adversarial attacks that corrupt or share faulty information during their training. In the distributed systems literature, the Byzantine model defines an omniscient adversary capable of injecting arbitrary errors or disruptions via hardware, software, and communication compromises~\cite{su2021byzantine, leblanc2013resilient}. Thus, many resilient algorithms have been studied to counter or contain the impacts of Byzantine agents in distributed consensus~\cite{leblanc2013resilient,diaji2019resilient,yuan2025resilient, ACC2026, CDC2025}, optimization~\cite{sundaram2019distributed, su2021byzantine, yemini2025resilient, zhai2025byzantine}, and learning frameworks~\cite{fang2022bridge, chen2017distributed, lin2024sf, blanchard2017machine}. Similarly, Byzantine-resilient distributed MARL has been studied in recent years. Early studies showed that even a single adversarial agent can severely disrupt the learning process in cooperative MARL algorithms~\cite{figura2021adversarial, xie2021towards}. It has been established in~\cite{hairi2024onthehardness} that exact evaluation of the honest agents' average reward is generally impossible under Byzantine attacks.

To address this vulnerability, several resilient MARL algorithms have been proposed. The QD-learning algorithm from~\cite{kar2013QD} was adapted in~\cite{xie2021towards,xie2023communication} to ensure convergence despite Byzantine agents. Other works have utilized linear function approximation to learn Q-values with adversarial attacks~\cite{wu2021byzantine_journal, yao2024communication_efficient, ye2024resilient}. However, these methods can at best guarantee convergence to near-optimal value functions, and therefore cannot guarantee learning of the optimal policies. Furthermore, many of these approaches require the communication network to satisfy specific robustness conditions (e.g., $(2F+1)$-robustness). Since verifying such properties is a co-NP-complete problem~\cite{zhang2015notion}, the application of these approaches to large-scale networks is limited. While~\cite{lin2020toward,lin2024robust,fang2025provably} proposed alternative algorithms that avoid such robustness conditions, they require a trusted central coordinator. 




In this work, we propose a fully resilient distributed Q-learning algorithm under which agents achieve almost sure convergence to \textit{optimal value functions} in the presence of Byzantine edge attacks, a restricted variant of the Byzantine model, in a decentralized network. By focusing on edge-level adversaries only, we provide stronger learning guarantees than prior work~\cite{xie2023communication, xie2021towards, zhai2025byzantine, wu2021byzantine_journal}, which achieve only near-optimal convergence or require additional structural assumptions on Q-values or local objectives. The key feature of our method is that it relies on redundant two-hop information to verify incoming messages and filter adversarial injections. We prove that under a novel topological condition, this approach ensures almost sure convergence to the optimal solution without any structural assumptions on the problem.



\textit{Contributions: } This paper introduces Fully Resilient Distributed Q-learning (FRQD-learning), a decentralized Q-learning algorithm that enables a network of agents to achieve almost sure asymptotic convergence to the optimal value functions despite Byzantine edge attacks. Specifically,
\begin{itemize}
    \item We integrate consensus-based distributed Q-learning (QD-learning)~\cite{kar2013QD} with a novel filtering mechanism that validates incoming information using two-hop messages, resulting in a Byzantine-resilient algorithm. We show that under a novel topological condition -- termed $(r,r')$-redundancy -- our method guarantees almost sure convergence to the optimal value functions.
    \item We present systematic constructions of $(r,r')$-redundant graphs and prove $(r,r')$-redundancy can be verified in polynomial time, unlike the co-NP-complete $r$-robustness conditions in~\cite{xie2023communication, yao2024communication_efficient}.
    \item We validate our results through simulations, showing that our method converges to the optimal solutions, whereas prior methods fail to do so under the same Byzantine edge attacks.
\end{itemize}

\section{Problem Statement}
A multiset $\Scal$ is a collection that may contain repeated elements. For simplicity, we abuse notation and use calligraphic symbols to denote both sets and multisets. We denote the cardinality of a (multi)set $\Scal$ as $|\Scal|$. We denote the sets of non-negative and positive integers as $\mathbb Z_{\geq 0}$ and $\mathbb Z_{>0}$. For a multiset $\Scal$ and $q \in \R$, we define ${\rm count}(q, \Scal)$ as the number of occurrences of $q$ in $\Scal$. We denote the probability and expectation of a probability space $(\Omega, \Fcal, \Pbb)$ by $\Pbb(\cdot)$ and $\Ebb(\cdot)$.

We consider a system of $n$ agents interacting over a simple, undirected, and time-varying communication graph $\Gcal(t)=(\Vcal, \Ecal(t))$. The vertex set $\Vcal=\{1,\dots, n\}$ represents the agents, and the edge set $\Ecal(t)\subset \Vcal \times \Vcal$ denotes a set of communication links between agents at time $t$. Since the graph is undirected, $(i,j)\in \Ecal(t)$ implies that $(j,i)\in \Ecal(t)$. For agent $i$, its one-hop and two-hop neighbor sets at time $t$ are denoted by $
\Ncal_i(t)=\{j\in \Vcal \mid (i,j)\in \Ecal(t)\}$ and 
$\Ncal_i^{(2)}(t)=\{j^{(2)}\in\Vcal \mid \exists\, j\in\Ncal_i(t)\ \text{s.t.}\  j^{(2)}\in \Ncal_j(t)\setminus \{i\}\}$. An extended one-hop neighbor set is $\Bcal_i(t) = \Ncal_i(t)\cup \{i\}$. A path of length $k\in\Zbb_{>0}$ is a sequence of vertices $(v_0, \dots, v_k)$ such that $(v_{i-1}, v_i)\in \Ecal(t)$ $\forall i\in\{1,\dots,k\}$. A graph is connected if there exists a path between any pair of nodes.

The environment is modeled as a networked multi-agent Markov Decision Process (MDP) defined by a tuple $(\Xcal, \Ucal, \Pbb, \{c^i\}_{i\in\Vcal}, \gamma)$. Here, $\Xcal=\{1,\dots, M\}$ is the finite state space, and $\Ucal$ is the finite set of control actions. The transition function $\Pbb:\Xcal \times \Ucal \times \Xcal \to [0,1]$ is given by 
\eqnN{
\Pbb(\xbf_{t+1}=x' \mid \xbf_t=x, \ubf_t = u)= p_{xx'}^{u}, \forall x,x' \in \Xcal, u \in \Ucal,
}
where $\sum_{x'\in \Xcal} p_{xx'}^{u}=1$ for all $x \in \Xcal$ and $u\in \Ucal$. Each agent $i$ has a private random local cost function $c^i(x,u)$, which represents the cost of action $u$ at state $x$. The global cost is denoted as $c(x, u)=(1/n)\sum_{i\in \Vcal} c^i(x,u)$ for all $x\in \Xcal$ and $u\in \Ucal$. The constant $\gamma\in(0,1)$ is the discount factor. 

For a stationary control policy $\pi:\Xcal\to \Ucal$, where $\ubf_t=\pi(\xbf_t)$, and initial state $x$, the state process $\{\xbf_t^{\pi}\}$ evolves as a homogeneous Markov chain with transitions $p_{xx'}^{\pi(x)}$.
In RL, the goal is to learn a policy $\pi:\Xcal \to \Ucal$ that minimizes the infinite-horizon discounted global cost
\eqn{
V_{x,\pi}=\Ebb\left[
\sum_{t=0}^{\infty} \gamma^t c(\xbf_t,\pi(\xbf_t))
\;\middle|\;
\xbf_0=x
\right].
}

One of the most widely used algorithms is the Q-learning algorithm~\cite{sutton1998reinforcement}, which estimates the optimal state-action value functions (or Q-values) $\Qbf^*=[Q^*_{x,u}] \in \R^{|\Xcal \times \Ucal|}$. These values are unique solutions to the Bellman optimality equation
\eqnN{
Q^*_{x,u}
=
\Ebb[c(x,u)]
+
\gamma
\sum_{x'\in\mathcal X}
p_{xx'}^{u}
\min_{v\in\mathcal U}
Q^*_{x',v},
}
where $V^*_x= \min_{u\in\mathcal U} Q^*_{x,u}$ is the optimal value function for state $x\in \Xcal$. We denote the vector of optimal value functions and optimal policies as
\eqn{\label{eq:optimal_value}
 \Vbf^* = \begin{bmatrix}
    V^*_1 & \cdots & V^*_M
\end{bmatrix}^\top,
 \ \
\pi^*(x)\in\arg\min_{u\in\mathcal U} Q^*_{x,u}.
}

In the decentralized setting, while each agent $i\in \Vcal$ observes the state $x \in \mathcal{X}$, action $u \in \Ucal$, and its local cost $c^i(x,u)$, it does not know the costs of others. Hence, agents cannot compute the global cost $c(x, u)=(1/n)\sum_{i\in \Vcal} c^i(x,u)$ and the optimal value functions
\eqn{\label{eq:individual_optimal_value}
V^*_x = \inf_{\pi}\frac 1 n \sum_{i \in \Vcal} V^i_{x,\pi}, \ \forall x \in \Xcal,}
where $V^i_{x,\pi}$ denotes the discounted cost of agent $i$ under policy $\pi$:
\eqn{V_{x,\pi}^i = \Ebb\left[ \sum_{t=0}^\infty\gamma^t c^i(\xbf_t, \pi(\xbf_t))\mid \xbf_0=x\right].}

In response,~\cite{kar2013QD} presented the QD-learning algorithm, a distributed variant of Q-learning, that allows each agent to learn the optimal value functions and policies~\eqref{eq:optimal_value} by exchanging its local Q-values. However, the QD-learning algorithm fails in the presence of Byzantine agents -- those who represent a broad range of possible faults or adversarial attacks on software, hardware, and communication -- as they may deviate arbitrarily from prescribed protocols~\cite{leblanc2013resilient,su2021byzantine}.

Although resilient Q-learning variants~\cite{xie2021towards, xie2023communication} have been proposed recently, they only converge to near-optimal value functions in general. In fact, finding the exact optimal value functions in the presence of Byzantine agents is generally impossible~\cite{hairi2024onthehardness}, making the learning of optimal policies a fundamental challenge. Hence, we focus on a slightly weaker but highly practical adversarial model:

\begin{definition}[$F$-total Byzantine Edge Attack] Let there be $C\in\Zbb_{>0}$ rounds of communications among agents at time $t\in \Zbb_{\geq 0}$. An edge $(i,j)\in \Ecal_{\Bcal}^z(t) \subset \Ecal(t)$ is said to be under a \textbf{Byzantine edge attack} during the communication round $z\in\{1,\dots, C\}$ at time $t$ if the message sent by agent $i$ at round $z$ is altered arbitrarily or dropped before it is received by agent $j$ (and/or vice versa). The network $\Gcal(t)=(\Vcal, \Ecal(t))$ is said to be under \textbf{$F$-total Byzantine edge attack} if $|\Ecal_{\Bcal}^z(t)|\leq F$ for all $z\in\{1,\dots, C\}$ and $t\in \mathbb Z_{\geq 0}$.
\label{def:byzantine_attack}
\end{definition}
By definition, we only consider \textit{unreliability in communications,} unlike many works that focus on scenarios where the nodes themselves are unreliable. Formally, we assume
\begin{assumption}
    All agents $i\in \Vcal$ are cooperative, i.e., they follow the same prescribed protocol. \label{assum:cooperative}
\end{assumption}

We assume that Byzantine edge attackers have limited resources and can compromise only $F$ communications \textit{per round of communications}. For any attacked edge $(i,j)$ during round $z\in\{1,\dots, C\}$ at time $t$, both agents $i$ and $j$ may receive corrupted information, affecting at most $2F$ agents each communication round. This attack type is similar to those studied in~\cite{bhushan2017man,lei2026distributed}, and by~\cite[Prop. 6]{xie2023communication}, this attack can prevent agents from finding the optimal value functions~\eqref{eq:optimal_value}. Therefore, our problem of interest is:
\begin{problem}
\label{prob:problem}
Design a fully resilient distributed algorithm such that, under Assumption \ref{assum:cooperative} and an $F$-total Byzantine edge attack, every agent $i \in \Vcal$ finds the optimal value functions $\Vbf^*$~\eqref{eq:optimal_value} using only local interactions in a time-varying communication graph $\Gcal(t)=(\Vcal, \Ecal(t))$.
\end{problem}


We emphasize that by adopting a slightly weaker adversarial model than prior works~\cite{xie2023communication, ye2024resilient, wu2021byzantine_journal}, we obtain strictly stronger performance guarantees, namely, almost sure convergence to the exact optimal value functions, under $F$-total Byzantine edge attacks. To this end, we propose the Fully Resilient QD (FRQD)-learning algorithm (\Cref{alg:fully_resilient}), which leverages two-hop information to validate messages through redundancy. We show that our method solves~\Cref{prob:problem} under a novel topological condition (which is given in~\Cref{thm:main}).

\section{Fully Resilient $QD$-Learning}

Before presenting our solution, we first present the QD-learning algorithm~\cite{kar2013QD}. Let each agent $i\in \Vcal$ maintain a sequence of state-action value functions $\{\Qbf^i(t)\}\in \R^{|\Xcal \times \Ucal|}$ with components $Q^i_{x,u}(t)$ for every possible state-action pair $(x,u)$. Each agent also maintains a sequence of state value functions $\{\Vbf^i(t)\}\in \R^{M}$ with components $V^i_x(t) = \min_{u\in \Ucal} Q^i_{x,u}(t)$ for $x\in \Xcal$. The sequence $\{Q^i_{x,u}(t)\}$ evolves in the following rule:
\eqn{\label{eq:q_update}
Q^i_{x,u}(t+1)=Q^i_{x,u}(t) -\beta_{x,u}(t)\sum_{Q^j\in \Pcal^i_N(t)} \left( Q^i_{x,u}(t)-Q^j\right)  \\
+ \alpha_{x,u}(t) \left( c^i(\xbf_t,\ubf_t) + \gamma \min_{v\in \Ucal} Q^i_{\xbf_{t+1},v}(t) - Q^i_{x,u}(t)\right)\notag,
}
with the multiset $\Pcal^i_N(t) =\{Q^j_{x,u}(t) \mid j \in \Ncal_i(t)\}$. The weights are defined as
\eqn{
\label{eq:alphas}
\alpha_{x,u}(t) = 
\begin{cases} 
\frac {a}{(k+1)^{\tau_1}} & \text{if } t = T_{x,u}(k) \text{ for some } k \ge 0, \\ 
0 & \text{otherwise,} 
\end{cases}
}

\eqn{
\label{eq:betas}
\beta_{x,u}(t) = 
\begin{cases} 
\frac {b}{(k+1)^{\tau_2}} & \text{if } t = T_{x,u}(k) \text{ for some } k \ge 0, \\ 
0 & \text{otherwise,} 
\end{cases}
}
where $T_{x,u}(k)$ denotes the $(k+1)$-th sampling instant of state-action pair $(x, u)$. We define $a,b>0$, $\tau_1\in(1/2,1]$, and $0<\tau_2<\tau_1-1/(2+\epsilon_1)$, where $\epsilon_1>0$ is a small constant governing the moment conditions of the local costs (see~\Cref{assum:MDP}). Note that $\alpha_{x,u}(t) = \beta_{x,u}(t) = 0$ unless
$(\mathbf{x}_t,\mathbf{u}_t) = (x,u)$, so $Q^i_{x,u}(t)$ is updated at time $t$
only when the current state-action pair is $(x,u)$.

It was shown in~\cite[Thm. 3.1]{kar2013QD} that all agents $i$ achieve $\Qbf^{i}(t) \to \Qbf^*$ and $ \Vbf^{i}(t) \to \Vbf^*$ almost surely as $t \to \infty$ under the protocol~\eqref{eq:q_update} and assumptions~\cite[M.1-M.5]{kar2013QD}. We list two of the assumptions below:

\begin{assumption}
\label{assum:MDP}
    The probability space $(\Omega, \Fcal, \Pbb)$ is a complete probability space with filtration $\{\Fcal_t\}$ given by $\Fcal_t=\sigma(\{\xbf_{\tau}, \ubf_{\tau}\}_{\tau \leq t}, \{c^i(\xbf_{\tau}, \ubf_{\tau})\}_{i\in \Vcal, \tau < t})$, where $\sigma(\cdot)$ denotes the smallest $\sigma$-algebra with respect to where the random objects are measurable. The conditional probability for the controlled transition of $\{\xbf_t\}$ is $\Pbb(\xbf_{t+1}=x' \mid \Fcal_t)=p^{\ubf_t}_{\xbf_t x'}$. For all $i\in \Vcal$, $\Ebb[c^i(\xbf_t,\ubf_t)\mid \Fcal_t]=\Ebb[c^i(\xbf_t, \ubf_t)\mid \xbf_t, \ubf_t]$, which is equal to $\Ebb[c^i(x, u)]$ on the event $\{\xbf_t=x, \ubf_t=u\}$. Lastly, one-stage costs possess super-quadratic moments, i.e., $\Ebb[(c^i(x,u))^{2+\epsilon_1}]<\infty$ for all $i\in \Vcal$, $x\in \Xcal$, and $u\in \Ucal$.
\end{assumption}

\begin{assumption}
\label{assum:infinite_often}
    For all state-action pairs $(x,u)\in\Xcal \times \Ucal$ and $k\in \Zbb_{\geq 0}$, $\Pbb(T_{x,u}(k)<\infty)=1$.  
\end{assumption}

Note Assumptions~\ref{assum:MDP} and~\ref{assum:infinite_often} are equivalent to~\cite[M.1 and M.4]{kar2013QD}. The former implies that the system is a controlled MDP evolving on a complete probability space, where all agents observe the history of past states, controls, and local costs up to time $t$. The requirement for super-quadratic moments on costs is used for the almost sure convergence of our algorithm. The latter requires all state-action pairs to be visited infinitely often, which is standard for Q-learning. 

 While Byzantine-resilient variants of QD-learning have been proposed~\cite{xie2021towards,xie2023communication}, these methods in general only guarantee convergence to the neighborhood of the optimal value functions. This is because the methods rely on each agent $i$ modifying $\Pcal^i_N(t)$ in~\eqref{eq:q_update} by filtering extreme neighbor values. Such filtering can induce asymmetric information exchange, effectively resulting in a directed communication network and introducing biases into the update. In such situation, convergence to the optimal value functions is generally not guaranteed (cf.~\cite[Remark 3]{xie2021towards}).

\subsection{Fully Resilient QD (FRQD)-Learning}
 
 The main novelty of our work lies in how each agent determines $\Pcal^i_N(t)$ in~\eqref{eq:q_update} such that the \textit{actual induced communication remains undirected}. This allows us to recover the original QD-learning method of~\cite{kar2013QD} over a different network. Unlike traditional resilient methods that filter extreme values, which inherently breaks the symmetry of information flow, our method utilizes two-hop redundancy-based filtering.

 From a local perspective, agents cannot directly identify compromised messages from one-hop neighbors. Therefore, our approach leverages the fact that agent $i$ at time $t$ receives messages from the same two-hop neighbor $k \in \Ncal_i^{(2)}(t)$ through multiple independent one-hop neighbors (paths). By comparing these redundantly relayed messages, agents can detect and filter out those corrupted by an $F$-total Byzantine edge attack before incorporating them into the update.
\begin{algorithm}
\SetKwInOut{Inputs}{Inputs}

\caption{Fully Resilient QD (FRQD)-Learning}
\label{alg:fully_resilient}
\Inputs{Parameter $F$, initial $\Qbf_0^i, \Vbf_0^i$}
      \tcp{Receive data}
    Receive state $\xbf_t$, action $\ubf_t$, and cost $c^i(\xbf_t, \ubf_t)$
    
   Initialize multisets $\Kcal^i(t) = \emptyset$, $\Ccal^i(t)=\emptyset$, $\Pcal^i(t)=\emptyset$
    
    \tcp{Exchange Q-values with neighbors:}
    Send $(Q^i_{\xbf_t,\ubf_t}(t),i)$ to and receive $(Q^{j}_{\xbf_t,\ubf_t}(t), j)$ from neighbors

    \tcp{First Filtering:}
    \For{every received $( Q^{j}_{\xbf_t,\ubf_t}(t), j)$}{
    \If{$j\neq i$ and $j$ appears exactly once among all messages}{
     $\Kcal^i(t) = \Kcal^i(t) \cup\{(Q^{j}_{\xbf_t,\ubf_t}(t), j)\}$
     }
     }
     \tcp{Relay Q-values of its neighbors:}
     Share $\Kcal^i(t)$ with neighbors, and add $\Kcal^i(t)$ and each received $\Kcal^j(t)$ with no repeated indices to $\Ccal^i(t)$
    
   \tcp{Second Filtering:}
    \For{$k = 1$ to $n$}{
        $\Lcal^i_k(t) =\{Q^k_{\xbf_t,\ubf_t}(t) \mid (Q^k_{\xbf_t,\ubf_t}(t),k)\in {\Ccal^i(t)}\}$
        
        \For{each distinct $q \in {\Lcal^i_k(t)}$}{
            \If{${\rm count}(q,{\Lcal^i_k(t)}) \ge 3F+1$}{
                ${\Pcal^i(t)} = {\Pcal^i(t)} \cup \{q\}$
            }
            }
        }
        
   \tcp{Update:}
    Update $Q^i_{\xbf_t,\ubf_t}(t+1)$ using~\eqref{eq:q_update} with $\Pcal^i_N(t)=\Pcal^i(t)$
    
     Update $V^i_{\xbf_t}(t+1) = \min_{u\in \Ucal} Q^i_{\xbf_t, u}(t+1)$

\end{algorithm}




In Fully Resilient QD (FRQD)-learning algorithm~(\Cref{alg:fully_resilient}), there are six main steps with two rounds of communications. At line 1, agent $i\in\Vcal$ first receives the data $\xbf_t$, $\ubf_t$, and $c^i(\xbf_t,\ubf_t)$. At line 3, agent $i$ broadcasts its Q-value $Q^i_{\xbf_t,\ubf_t}(t)$ together with its index $i$ as a tuple $(Q^i_{\xbf_t,\ubf_t}(t),i)$ to its neighbors $j\in \Ncal_i(t)$ while receiving $(Q^{j}_{\xbf_t,\ubf_t}(t),j)$ from them. During the first filtering step (lines 4–6), agent $i$ discards any tuple whose index is equal to $i$ or appears more than once among the received messages. Thus, only tuples corresponding to indices that appear exactly once are retained in $\Kcal^i(t)$. The agent then shares $\Kcal^i(t)$ with its neighbors, allowing the propagation of Q-values from two-hop neighbors (line 7). Also, by lines 4-6, any shared $\Kcal^j(t)$ with duplicate agent indices must be compromised and thus not considered. In the second filtering step (lines 8–12), agent $i$ examines all relayed Q-values associated with agent $k\in \Vcal$; it considers a value trustworthy if it appears at least $3F+1$ times and puts them into the multiset $\Pcal^i(t)$. Finally, the values in $\Pcal^i(t)$ are used to update $Q^i_{\xbf_t,\ubf_t}(t)$ according to~\eqref{eq:q_update}, after which the value function is updated (lines 13–14).

\begin{remark}
The FRQD-learning algorithm requires agents to relay their own and their neighbors' Q-values over two communication rounds. In the first round, each agent sends one scalar to each neighbor ($O(|\Ncal_i(t)|)$ communication). In the second round, each agent relays up to $|\Ncal_i(t)|$ scalars received from its neighbors, resulting in $O(|\Ncal_i(t)|^2)$ communication at worst. Reducing this overhead while preserving the same convergence guarantees (\Cref{thm:main}) remains future work. \hfill $\bullet$
\end{remark}

\begin{remark}
Our algorithm uses multi-path redundancy for \emph{validation} and directly incorporates validated two-hop messages into the update. This differs from works such as~\cite{yuan2025resilient,hajicotis2025trustworthy}, which use two-hop messages only to \emph{detect and isolate} malicious agents from updates. Moreover, our approach provides resilience against point-to-point Byzantine communication attacks, whereas these works consider adversaries that only broadcast the same values to all neighbors. \hfill $\bullet$
\end{remark}

\subsection{$(r,r')$-redundancy}

Now we present novel topological properties under which the FRQD-learning algorithm solves~\Cref{prob:problem}. First, we define:
\begin{definition}[$r$-2-hop Graph]
Let $\Gcal(t) = (\Vcal, \Ecal(t))$ be an undirected graph at time $t$. We define the \textbf{$r$-2-hop graph of $\Gcal(t)$}, denoted as $\Gcal_r^{(2)}(t) = (\Vcal, \Ecal_r^{(2)}(t))$, such that an edge $(i, j) \in \Ecal_r^{(2)}(t)$ if $|\Bcal_i(t)\cap \Ncal_j(t)|\geq r$.
\end{definition}

An illustrative example of an $r$-2-hop graph is given in~\Cref{fig:illustrative_example}. An $r$-2-hop graph contains an edge $(i,j)$ if agents $i$ and $j$ share at least $r$ neighbors (including themselves). Equivalently, there are at least $r$ vertex-disjoint paths of length at most 2 connecting them. 
\begin{figure}
    \centering
\includegraphics[width=0.9\linewidth]{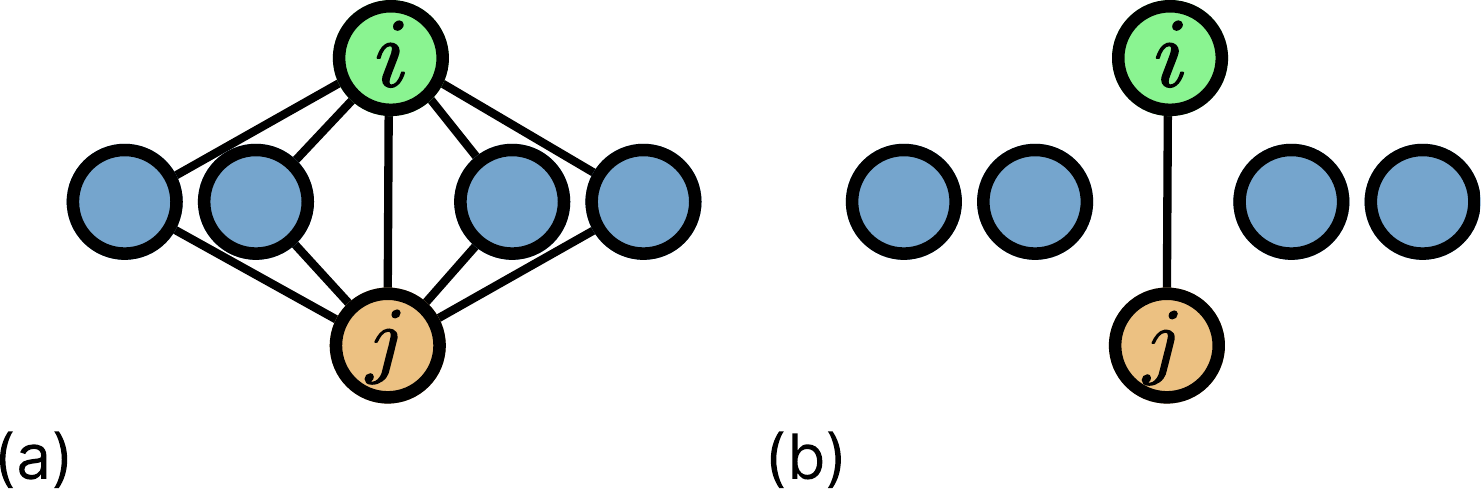}
    \caption{Visualizations of (a) graph $\Gcal(t)=(\Vcal, \Ecal(t))$ and (b) its $5$-2-hop graph $\Gcal^{(2)}_5(t)=(\Vcal, \Ecal^{(2)}_5(t))$ at time $t$. The edge $(i,j)\in \Ecal^{(2)}_5(t)$, since $|\Bcal_i(t)\cap \Ncal_{j}(t)|\geq 5$.}
    \label{fig:illustrative_example}
\end{figure}
\begin{figure}
    \centering
\includegraphics[width=0.9\linewidth]{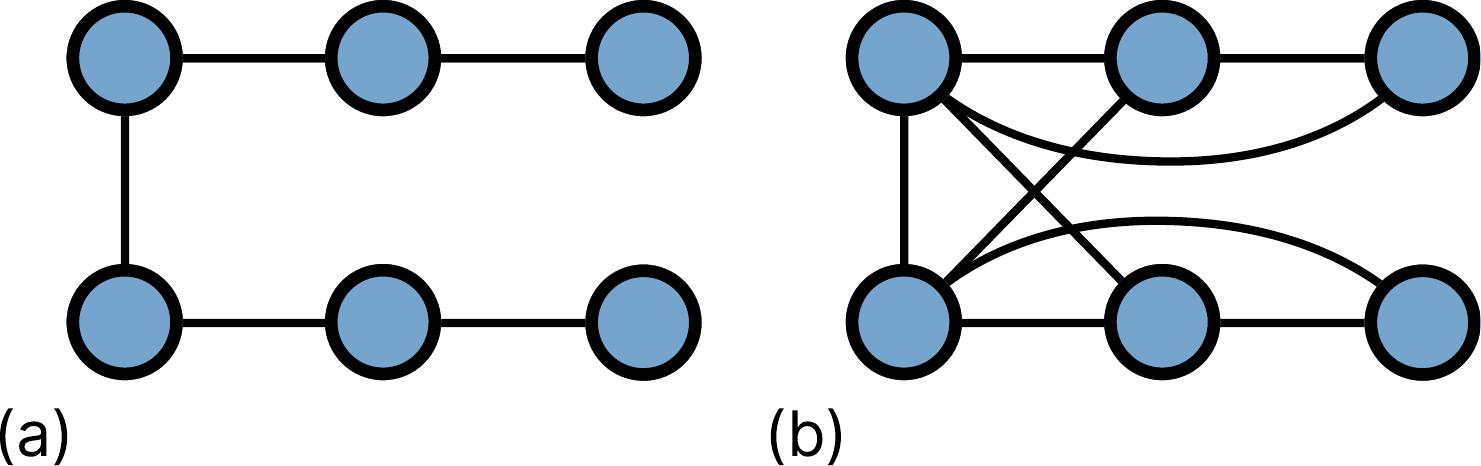}
    \caption{(a) $(1,0)$-redundant graph and (b) its 1-2-hop graph.}
    \label{fig:example_of_transformation}
\end{figure}

\begin{definition}[$(r,r')$-redundant]
\label{def:redundant}
Let $\Gcal(t) = (\Vcal, \Ecal(t))$ be an undirected graph at time $t$, and let 
$\Gcal_r^{(2)}(t) = (\Vcal, \Ecal_r^{(2)}(t))$ be its $r$-2-hop graph. We say that $\Gcal(t)$ is \textbf{$(r,r')$-redundant} with $r>r'\geq 0$ at time $t$ if:
\begin{enumerate}
    \item $\Gcal_r^{(2)}(t)$ is connected, and
    \item for all $(i,j) \notin \Ecal_r^{(2)}(t)$, $|\Bcal_i(t)\cap \Ncal_j(t)| \le r'$.
\end{enumerate}
\end{definition}

\Cref{fig:example_of_transformation} visualizes an $(r,r')$-redundant graph and its $r$-2-hop graph. A graph $\Gcal(t)$ is $(r,r')$-redundant if two things hold. First, its $r$-2-hop graph is connected. Second, for all agent pairs not connected in the $r$-2-hop graph, they share at most $r'$ neighbors (including themselves) in the graph $\Gcal(t)$. That is, for any $i,j\in \Vcal$, they share either at least $r$ or at most $r'$ neighbors (including themselves). 

\begin{remark}
The $(r,r')$-redundancy condition quantifies the topological condition required for agents to have symmetric validations of correct messages amidst adversarial communication attacks. The threshold $r$ ensures that for any two agents in the $r$-2-hop graph, there is sufficient multi-path redundancy (at least $r$ common neighbors) to verify the relayed messages. 

Conversely, when two agents are not connected in the $r$-2-hop graph, the number of neighbors they share is bounded by $r'$. As a result, this prevents Byzantine edge attacks from manipulating relayed messages to create asymmetric validation (e.g., agent $i$ having sufficient redundancy in messages from agent $j$ while the reverse does not hold). This bound therefore is intended to ensure the symmetry of the underlying communication, ensuring convergence to the optimal value functions as described in~\eqref{eq:optimal_value}-\eqref{eq:individual_optimal_value}. ~\hfill $\bullet$
\end{remark}

\begin{lemma}
Let $\Gcal(t)=(\Vcal, \Ecal(t))$ be $(r,r')$-redundant at time $t$. Then, its $r$-2-hop graph $\Gcal^{(2)}_{r}(t)=(\Vcal, \Ecal^{(2)}_r(t))$ is connected and undirected at time $t$.
\label{lem:undirected}
\end{lemma}
\begin{proof}
By~\Cref{def:redundant}, $\Gcal^{(2)}_{r}(t)$ is connected. Furthermore, because $\Gcal(t)$ is undirected, $|\Bcal_i(t)\cap \Ncal_j(t)|=|\Bcal_j(t)\cap \Ncal_i(t)|$ for any $i,j\in \Vcal$, $i\neq j$. Hence, $\Gcal^{(2)}_{r}$ is also undirected.
\end{proof}

\subsection{Sufficient Conditions for the FRQD-Learning}

Before stating our main result (\Cref{thm:main}), we first look at the supporting lemmas.

\begin{lemma}
\label{lem:corrupted_values}
Let Assumption~\ref{assum:cooperative}  hold.  
Let $\Gcal(t)=(\Vcal, \Ecal(t))$ be connected, and suppose that each agent $i\in \Vcal$ runs the FRQD-learning algorithm under an $F$-total Byzantine edge attack. Then, for any time $t\in \Zbb_{\geq 0}$, agent $i \in \Vcal$, and $k \in \Vcal \setminus \{i\}$, the multiset $\Lcal^i_k(t)$ contains at most $3F$ values that are not equal to $Q^k_{\xbf_t, \ubf_t}(t)$, i.e.,
\eqn{\label{eq:bound}
\vert\{\tilde Q \in \Lcal^i_k(t) \mid \tilde Q \neq Q^k_{\xbf_t, \ubf_t}(t)\}\vert \le 3F.
}
\end{lemma}

\begin{proof}
By~\Cref{assum:cooperative}, we know all agents $i\in \Vcal$ share the true messages $(Q^i_{\xbf_t, \ubf_t}(t), i)$ and $\Kcal^i(t)$ at lines 3 and 7 of~\Cref{alg:fully_resilient}. Therefore, any messages with values $\tilde Q \neq Q^k_{\xbf_t, \ubf_t}(t)$ must have been corrupted by Byzantine attack. 

    At line 3, every agent receives Q-values directly from its neighbors. Since the $F$-total Byzantine edge attack is bidirectional, at most $2F$ agents receive a corrupted value $(\tilde Q, k)$ where $\tilde Q \neq Q^k_{\xbf_t, \ubf_t}(t)$. At line 7, each agent relays its neighbors' Q-values it has received to other neighbors. In the worst case, (i) all $2F$ agents who received $(\tilde Q, k)$ from the previous communication round have agent $i$ as their neighbors, and (ii) $F$ additional corrupted  $(\tilde Q, k)$ get relayed to agent $i$. Therefore, for any $k\neq i$, the set $\Lcal^i_k(t)$ contains at most $3F$ values that do not equal to $Q^k_{\xbf_t, \ubf_t}(t)$.
\end{proof}

The lemma establishes that an agent receives at most $3F$ corrupted values for every other agent as a result of the two rounds of communication under $F$-total Byzantine edge attacks (lines 3-7). The bound~\eqref{eq:bound} then guarantees that any value from a specific neighbor appearing at least $3F+1$ times must be correct (and thus line 11 in~\Cref{alg:fully_resilient}). 

Now, we present another lemma. Let $\Scal_{\xbf_t,\ubf_t}(t)=\{Q^1_{\xbf_t, \ubf_t}(t),\dots, Q^n_{\xbf_t, \ubf_t}(t)\}$ be the set of Q-values for all agents at state $\xbf_t$ and $\ubf_t$ at time $t$. 

\begin{lemma}
\label{lem:laplacian_connection}
Let~\Cref{assum:cooperative} hold.
    Let $\Gcal(t)=(\Vcal, \Ecal(t))$ be $(6F+1, 0)$-redundant, and let $\Gcal_{6F+1}^{(2)}(t)$ be the $(6F+1)$-2-hop graph of $\Gcal(t)$ at time $t$. Suppose that each agent $i\in \Vcal$ runs the FRQD-learning algorithm under an $F$-total Byzantine edge attack. Then, the following holds:
\eqn{\label{eq:Q_equivalent_update}
    \Qbf_{x,u}(t+1) & = \begin{bmatrix}
        Q_{x,u}^1(t+1) & \cdots &     Q_{x,u}^n(t+1)
    \end{bmatrix}^\top \\ & = \left((1-\alpha_{x,u}(t))I_n - \beta_{x,u}(t) L_t\right)\Qbf_{x,u}(t)  \nonumber\\ & \ \ 
    +\alpha_{x,u}(t)\nu_{\xbf_t,\ubf_t}(t)\nonumber,}
where $L_t := L(\Gcal_{6F+1}^{(2)}(t))$ is the Laplacian matrix of $\Gcal_{6F+1}^{(2)}(t)$ and $\nu_{\xbf_t,\ubf_t}(t)=\begin{bmatrix}
    \nu_{\xbf_t,\ubf_t}^1(t) & \cdots &   \nu_{\xbf_t,\ubf_t}^n(t) 
\end{bmatrix}^\top$ with $\nu^i_{\xbf_t,\ubf_t}(t) = c^i(\xbf_t,\ubf_t) + \gamma \min_{v\in \Ucal} Q^i_{\xbf_{t+1},v}(t)$.
    
\end{lemma}
\begin{proof}

Our proof is in three parts: (i) we show that $(i,j)\in \Ecal^{(2)}_{6F+1}(t)$ if and only if $Q^i_{\xbf_t,\ubf_t}(t)\in \Pcal^{j}(t)$ and $Q^{j}_{\xbf_t,\ubf_t}(t)\in \Pcal^{i}(t)$. (ii) We prove that $\Pcal^i(t)\subseteq \Scal_{\xbf_t,\ubf_t}(t)$ for all $i\in \Vcal$ and $t\in \Zbb_{\geq 0}$. (iii) We combine (i) and (ii) to complete the proof.

\textbf{Part 1a $(\implies)$:} Let $(i,j)\in\Ecal^{(2)}_{6F+1}(t)$. Then, $|\Ncal_i(t)\cap\Bcal_j(t)|\geq 6F+1$ by definition and \Cref{lem:undirected}.

By~\Cref{assum:cooperative}, each agent $i\in \Vcal$ correctly shares $(Q^i_{\xbf_t, \ubf_t}(t), i)$ and $\Kcal^i(t)$ at lines 3 and 7 of~\Cref{alg:fully_resilient}, respectively. By definition of $F$-total Byzantine edge attack (where each edge attack affects at most two agents), at least $4F+1$ nodes $j'\in \Ncal_i(t)\cap \Bcal_{j}(t)$ will correctly have $(Q^i_{\xbf_t, \ubf_t}(t), i)$ in $\Kcal^{j'}(t)$ by line 6. Now, (i) because each agent $j'$ correctly shares $\Kcal^{j'}(t)$ at line 7 by~\Cref{assum:cooperative} and (ii) by definition of $F$-total Byzantine edge attack, agent $j$ will have at least $3F+1$ repetitions of $(Q^i_{\xbf_t, \ubf_t}(t), i)$ in $\Ccal^{j}(t)$. Hence, we have ${\rm count}(Q^i_{\xbf_t, \ubf_t}(t), \Lcal^{j}_i(t))\geq 3F+1$, which means $Q^i_{\xbf_t, \ubf_t}(t)\in \Pcal^{j}(t)$. 

Since $\Gcal^{(2)}_{6F+1}$ is undirected by~\Cref{lem:undirected}, we conclude that $Q^j_{\xbf_t, \ubf_t}(t)$ will be in $\Pcal^{i}(t)$ also. Finally, since an edge $(i',j')$ exists in $\Ecal^{(2)}_{6F+1}(t)$ implies $|\Ncal_{i'}(t)\cap \Bcal_{j'}(t)|\geq 6F+1$ by definition, we conclude for any $(i,j)\in \Ecal^{(2)}_{6F+1}(t)$, $Q^i_{\xbf_t,\ubf_t}(t)\in \Pcal^{j}(t)$ and $Q^{j}_{\xbf_t,\ubf_t}(t)\in \Pcal^{i}(t)$.

\textbf{Part 1b $(\Longleftarrow)$:} Assume to the contrary that $Q^i_{\xbf_t,\ubf_t}(t)\in \Pcal^{j}(t)$ and $Q^{j}_{\xbf_t,\ubf_t}(t)\in \Pcal^{i}(t)$ but also $(i,j)\notin \Ecal^{(2)}_{6F+1}(t)$. For $Q^i_{\xbf_t,\ubf_t}(t)$ and $Q^{j}_{\xbf_t,\ubf_t}(t)$ to be in $\Pcal^{j}(t)$ and $\Pcal^{i}(t)$, respectively, ${\rm count}(Q^i_{\xbf_t,\ubf_t}(t),\Lcal^j_i(t))\geq 3F+1$ and ${\rm count}(Q^j_{\xbf_t,\ubf_t}(t),\Lcal^i_j(t))\geq 3F+1$. Nevertheless, $(i,j)\notin \Ecal^{(2)}_{6F+1}(t)$ implies that $|\Ncal_i(t)\cap \Bcal_{j}(t)|<6F+1$ and thus $|\Ncal_i(t)\cap \Bcal_{j}(t)|=0$ by definition. Since Byzantine edge attacks can manipulate at most $3F$ values in $\Lcal^j_i(t)$ and $\Lcal^i_j(t)$ by~\Cref{lem:corrupted_values}, we reach a contradiction.

\textbf{Part 2:} Assume to the contrary that there exist $i \in \Vcal$ and $t\in \Zbb_{\geq 0}$ such that $\Pcal^i(t) \not\subseteq \Scal_{\xbf_t, \ubf_t}(t)$. This implies there exists $Q'\in \Pcal^i(t)$ such that $Q' \notin \Scal_{\xbf_t, \ubf_t}(t)$. For this to happen, ${\rm count}(Q',\Lcal^i_k(t))\geq 3F+1$ has to hold for some $k\in \Vcal\setminus \{i\}$. This is a contradiction, as Byzantine edge attacks can manipulate at most $3F$ values $Q' \notin \Scal_{\xbf_t, \ubf_t}(t)$ by~\Cref{lem:corrupted_values}.

\textbf{Part 3:} From Part 1, the existence of an edge $(i,j)$ in $\Ecal^{(2)}_{6F+1}(t)$ is necessary and sufficient for $Q^i_{\xbf_t,\ubf_t}(t) \in \Pcal^j(t)$ and $Q^j_{\xbf_t,\ubf_t}(t)\in \Pcal^i(t)$. In addition, from Part 2, we know that each agent $i\in \Vcal$ uses values only in $\Scal_{\xbf_t, \ubf_t}(t)$ to update its $Q^i_{\xbf_t, \ubf_t}(t+1)$ for all time $t$. Therefore, this is equivalent to agents $i\in \Vcal$ sharing $Q^i_{\xbf_t, \ubf_t}(t)$ to agents in $\Ncal^{(2)}_{i,6F+1}(t)=\{j\in \Vcal \mid (i,j)\in \Ecal^{(2)}_{6F+1}(t)\}$ before update. Hence, the update step from~\eqref{eq:q_update} with $\Pcal^i_N(t)=\Pcal^i(t)$ can be rewritten as

\eqnN{
Q^i_{x,u}(t+1)= &(1-\alpha_{x,u}(t))Q^i_{x,u}(t) + \alpha_{x,u}(t) \nu^i_{\xbf_t,\ubf_t}(t) \\ &-\beta_{x,u}(t)\sum_{j\in \Ncal^{(2)}_{i,6F+1}(t)} \left( Q^i_{x,u}(t)-Q^j_{x,u}(t)\right),
}
for each $x\in \Xcal$ and $u\in \Ucal$, which completes the proof.
\end{proof}

\Cref{lem:laplacian_connection} establishes that running the FRQD-learning algorithm on $\Gcal(t)$ is mathematically equivalent to executing the standard QD-learning algorithm~\cite{kar2013QD} on the $(6F+1)$-2-hop graph $\Gcal_{6F+1}^{(2)}(t)$ of $\Gcal(t)$ at time step $t$. Now, we present the main result of our paper:
\begin{theorem}
\label{thm:main}
    Let Assumptions~\ref{assum:cooperative}-\ref{assum:infinite_often} hold. If, for all $t \in\Zbb_{\geq 0}$, (i) $\Gcal(t)=(\Vcal, \Ecal(t))$ is $(6F+1, 0)$-redundant under an $F$-total Byzantine edge attack, and (ii) each agent $i\in \Vcal$ runs the FRQD-learning algorithm, then
    for each agent $i\in \Vcal$, 
    \eqn{\Pbb\left(\lim_{t\to \infty} \Qbf^i(t) = \Qbf^*\right)=1, \\
    \Pbb\left(\lim_{t\to \infty} \Vbf^{i}(t) =  \Vbf^*\right)=1.
    }
\end{theorem}

\begin{proof}
The proof proceeds by combining the result from~\Cref{lem:laplacian_connection} with the analysis in~\cite[Thm.~3.1]{kar2013QD}.

By~\Cref{lem:laplacian_connection}, we know
the update rule of FRQD-learning can be written as~\eqref{eq:Q_equivalent_update}, where $L_t := L(\Gcal_{6F+1}^{(2)}(t))$, which is in the same update form as that studied in~\cite[Thm. 3.1]{kar2013QD}. 

It remains to verify that the assumptions (M.1)--(M.5) of
\cite[Thm.~3.1]{kar2013QD} are satisfied. \Cref{assum:MDP} is equivalent
to (M.1), while~\Cref{assum:infinite_often} is equivalent to (M.4). Since
$L_t$ is determined by $\mathcal{G}(t)$ alone and is unaffected by the
attack (\Cref{lem:laplacian_connection}), it is deterministic, satisfying the independence condition (M.3).
Furthermore, by~\Cref{lem:undirected}, the $(6F+1)$-2-hop graph is connected
and undirected for all $t$, so $\lambda_2(L_t)>0$. Since $\mathcal{V}$ is
fixed and finite, $\{L_t\}$ takes values in a finite set of Laplacians, and
hence $\inf_{t}\lambda_2(L_t)>0$,
which is what (M.2) requires of a deterministic Laplacian sequence. Finally, the step-sizes~\eqref{eq:alphas}-\eqref{eq:betas}
satisfy the persistence and two-time-scale requirements in (M.5). Hence, all
assumptions of~\cite[Thm.~3.1]{kar2013QD} hold, completing the proof.
\end{proof}


\begin{remark}
\Cref{thm:main} guarantees that, despite Byzantine edge attacks, all agents converge almost surely to the optimal value functions. This represents a notable improvement over prior works, which in general only guarantee convergence to near-optimal solutions. However, we would like to highlight that our stronger results depend on the fact that we are considering a weaker Byzantine attack model (\Cref{def:byzantine_attack}), that is, manipulation happens only in the communication, and all agents follow the prescribed protocol~(\Cref{assum:cooperative}).
\end{remark}

\section{Construction of Redundant Network Graph}

In the previous section, we identified $(6F+1,0)$-redundancy as a sufficient condition for all agents to learn the optimal value functions $\Vbf^*$~\eqref{eq:optimal_value}. In this section, we present (i) a systematic construction of an $(r,r')$-redundant graph (\Cref{prop:construction}) and (ii) the complexity of verifying it (\Cref{prop:computation}). As the underlying structural requirements are independent of the learning dynamics, we drop the argument $t$ on a graph throughout this discussion.

We first present a systematic method to construct $(r,r')$-redundant graphs for any $r$ and $r'$:

\begin{prop}
\label{prop:construction}
    Let $\Vcal=\{1,\dots, n\}$ where $n>r$, and $\Vcal_c=\{1,\dots, r\}\subset \Vcal$. Then, a graph $\Gcal=(\Vcal, \Ecal)$ is $(r,r')$-redundant for $r>r'\geq 0$ if (i) every node in $\Vcal_c$ is connected to every other node in $\Vcal_c$ and (ii) each node $i\in\Vcal\setminus \Vcal_c$ is connected to all $r$ nodes in $\Vcal_c$.
\end{prop}
\begin{proof}

Here we prove by induction that $\Gcal$ is $(r,r')$-redundant for any $n\geq r+1$.

When $n=k=r+1$, every node in $\Vcal$ is connected to each other. Then, for any $i,j\in \Vcal$ such that $i\neq j$, $|\Bcal_i \cap \Ncal_j|\geq r$. Hence, $r$-2-hop graph $\Gcal_r^{(2)}=(\Vcal, \Ecal_r^{(2)})$ of $\Gcal$ is fully connected. Because $|\Bcal_i \cap \Ncal_j|\geq r$ holds for any $i,j\in \Vcal$, $\Gcal$ is $(r,r')$-redundant for any $r'\geq 0$.

Now, assume that $\Gcal$ is $(r,r')$-redundant for any $r'\geq 0$ with $n=k>r+1$. By assumption, $\Gcal$ is $(r,r')$-redundant for any $r'\geq 0$ with $n=k>r+1$, which implies $|\Bcal_i \cap \Ncal_j|\geq r$ for any $i,j\in \Vcal_c\cup \{r+1,\dots, k\}$. Then, if $n=k+1$, $\Gcal$ is constructed by having agent $k+1$ connected to all $r$ nodes in $\Vcal_c$ which are also connected to agents $r+2,\dots, k$. Hence, agent $k+1$ also has $|\Bcal_{k+1}\cap \Ncal_i|\geq r$ for any $i\in \Vcal_c\cup\{r+1,\dots, k\}$, implying $\Gcal$ is $(r,r')$-redundant. 
\end{proof}

While \Cref{prop:construction} provides a method to construct $(r,r')$-redundant graphs, the following result establishes that we can verify the redundancy of an arbitrary graph efficiently.

\begin{prop}
Given $r, r'\in \Zbb_{\geq 0}$ and a communication graph $\Gcal=(\Vcal, \Ecal)$ with $|\Vcal|=n$, one can verify whether $\Gcal$ is $(r,r')$-redundant in $O(n^{3})$.
    \label{prop:computation}
\end{prop}
\begin{proof}
    Let $A$ be an adjacency matrix of $\Gcal$. Then, $\bar{A}:=A^2+A$ will contain elements $\bar{a}_{ij}$ that counts the number of shared neighbors between nodes $i$ and $j$ (including node $i$ itself) i.e., $|\Bcal_i\cap \Ncal_j|$. Computing $\bar{A}$ using standard matrix multiplication requires $O(n^3)$ operations~\cite[Sec. 4]{cormen2009intro_to_alg}, and checking all entries adds at most $O(n^2)$ operations. Next, to verify that the $r$-2-hop graph $\Gcal_r^{(2)} = (\Vcal, \Ecal_r^{(2)})$ of $\Gcal$ is connected, one can perform a Breadth-First Search (BFS), which requires $O(n + m_r)$ time, where $m_r = |\Ecal_r^{(2)}|$~\cite[Sec. 22.2]{cormen2009intro_to_alg}. Therefore, the total required computation is $O(n^{3}+m_r)$. Since $m_r\leq \binom{n}{2}=O(n^2)$, $O(n^{3}+m_r)=O(n^{3})$.
\end{proof}

\Cref{prop:computation} establishes that $(r,r')$-redundancy can be verified efficiently. Compare this with $r$-robustness~\cite{leblanc2013resilient}, whose definition is given below: 
\begin{definition}[$\mathbf r$\textbf{-robustness}~\cite{leblanc2013resilient}]
    A graph $\Gcal = (\Vcal,\Ecal)$ is $\mathbf r$\textbf{-robust} if for every pair of nonempty, disjoint subsets $\Scal_1,\Scal_2 \subset \Vcal$, at least one of the subsets contains a node with at least $r$ neighbors outside itself. That is, there exists a node $i \in \Scal_k$ such that $|\Ncal_i \setminus \Scal_k| \geq r$ for some $k \in \{1, 2\}$.
    \label{def:r_robust}
\end{definition}

While $(2F+1)$-robustness provides a sufficient condition for Byzantine-resilient QD-learning in other MARL frameworks~\cite{xie2021towards, xie2023communication, yao2024communication_efficient}, determining whether a graph satisfies this property is co-NP-complete and thus computationally expensive~\cite{zhang2015notion, TAC2025}. This computational bottleneck renders robustness-based design impractical for large-scale and dynamic networks. In contrast, $(r,r')$-redundancy offers a tractable alternative, making our approach more suitable.
\begin{remark}
\label{remark:relationship_w_robust}
Studying the precise relationship between $(r,r')$-redundancy and $r$-robustness is an interesting direction for future work. Nevertheless, we note that the graph constructed in~\Cref{prop:construction} is at least $\lceil \frac{r+1}{2}\rceil$-robust. The key observation is that the construction contains a fully connected subgraph of size $r+1$, with $r$ of these nodes additionally connected to the remaining nodes. Since a complete graph with $n$ nodes is $\lceil \frac{n}{2}\rceil$-robust~\cite[Lem. 4]{leblanc2013resilient}, and robustness is preserved under the preferential attachment~\cite[Thm. 5]{leblanc2013resilient}, the claim follows. A formal proof is omitted due to space limitations. ~\hfill $\bullet$
\end{remark}

\section{Simulations}
\label{sec:sim}
Here, we validate our method through simulations. Consider a system of $n=10$ agents ($\Vcal=\{1, \dots, 10\}$) with heterogeneous capabilities interacting over a network. There are six different tasks that must be completed sequentially, each requiring a pair of robots. Since robots have different capabilities, their operation costs and performance (probability of completing the task) depend on both the task and assigned robot-pairs. There is no limit on the number of tasks in which each robot may participate.

To determine an optimal assignment policy that minimizes the total operation cost, we define the state space $\Xcal=\{1,\dots, 7\}$, where states $1,\dots,6$ correspond to the current task and state $7$ represents task completion. The control space is defined as $\Ucal=\{(i,j)\mid i,j\in \Vcal , i\neq j\}$, representing the deployment of robot pair. The system starts at the initial state $\xbf_0=1$ and the terminal state is $7$. The controlled transition distributions are defined as
\eqn{p_{xx'}^{(i,j)} = \begin{cases}
\label{eq:transition_dynamics}
    \frac {\abs{i-j}}{\abs{i-j}+x} & \text{ if } x' = x+1,\\
     \frac{x}{\abs{i-j}+x} & \text{ if } x' = x,\\
\end{cases}}
and $p_{77}^{(i,j)}=1$ for all $(i,j)\in \Ucal$.
The local cost function for agent $i\in \Vcal$ is defined as 
\eqn{
\label{eq:costs}
c^i(x,u) = \begin{cases}
\delta_{ij}(x)  & \text{ if } u=(i,j),\\
0 & \text{ otherwise},
\end{cases}
}
where $\delta_{ij}(x)$ is sampled uniformly from the interval $[0,50]$ at $t=0$ and remains fixed for $t>0$, with $c^i(7,u)=0$ for all $u\in \Ucal$. Notably, both the transition dynamics~\eqref{eq:transition_dynamics} and the cost structures~\eqref{eq:costs} are simultaneously state- and pair-dependent. We set the discount factor $\gamma=0.9$. The algorithm parameters are chosen as $\epsilon_1=10^{-4}$, $a=b=1/n$, $\tau_1=1$, and $\tau_2=\tau_1-1/(2+\epsilon_1)-\epsilon_2$, where $\epsilon_2=10^{-4}$.

We compare our method against two algorithms:
\begin{itemize}
    \item \textbf{Oracle:} The vanilla QD-learning algorithm~\cite{kar2013QD} in the absence of adversarial attacks. This serves as the ground-truth optimal value functions and policies.
    \item \textbf{Baseline:} The resilient QD-learning method proposed in~\cite{xie2021towards, xie2023communication} without the event-triggering mechanism.
\end{itemize} 

In all algorithms, agents interact over a $(7,0)$-redundant communication network $\Gcal=(\Vcal, \Ecal)$ constructed according to~\Cref{prop:construction}, which is also $4$-robust (see~\Cref{remark:relationship_w_robust}). Although the simulations are performed on a fixed graph as the Baseline only supports static networks, note our method naturally generalizes to time-varying graphs. We assume the network is under $F=1$-total Byzantine edge attack. \Cref{thm:main} guarantees that all agents almost surely converge to the optimal value functions using our FRQD-learning algorithm. In contrast, Baseline only ensures convergence to a neighborhood of the optimal value functions~\cite[Thm.~1]{xie2023communication}.

To simulate a Byzantine edge attack, at each communication round we randomly select $F=1$ edge and designate it as adversarial, with its transmitted information being corrupted. Specifically for our algorithm, during the first communication round (line 3 of~\Cref{alg:fully_resilient}), the adversary replaces the transmitted message with an extreme Q-value for agent $1$, namely $(10000,1)$. During the second communication round (line 7), the adversary injects the falsified set $\tilde{\Kcal} = \{(10000,i)\}_{i \in \mathcal V}$. In contrast, when simulating the Baseline, agents receiving information via the compromised edge receive the value $10000$ in place of true neighbor values.

\begin{figure}
    \centering
    \includegraphics[width=1\linewidth]{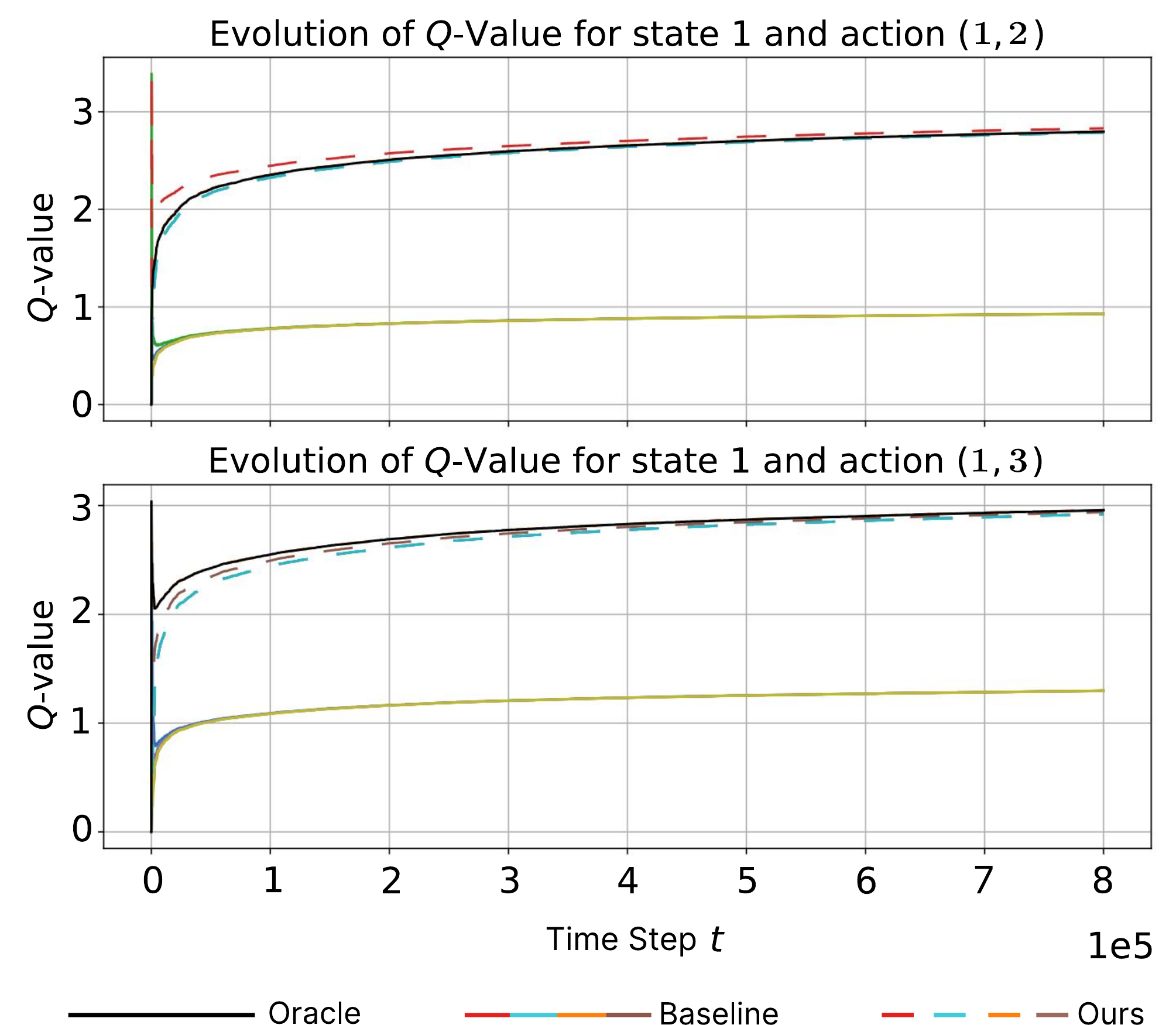}
    \caption{Performance comparison of our algorithm against the optimal Q-values from the Oracle~\cite{kar2013QD} (no attack) and the Baseline~\cite{xie2021towards,xie2023communication} under a $1$-total Byzantine edge attack.}
    \label{fig:simulation}
\end{figure}

\begin{table}[t]
\vspace{-10pt}
\small
\setlength{\tabcolsep}{3pt} 
\centering
\caption{Learned policies from each algorithm across states.}
\label{tab:optimal_policy_transposed}
\begin{tabular}{c|>{\columncolor{gray!20}}c c c>{\columncolor{gray!20}}c >{\columncolor{gray!20}}c >{\columncolor{gray!20}}c}
\toprule
 & $x=1$ & $x=2$ & $x=3$ & $x=4$ & $x=5$ & $x=6$ \\
\midrule
Oracle & $(1,2)$ & $(2,3)$ & $(4,7)$ & $(1,4)$ & $(1,5)$  & $(3,4)$ \\ \hline
Baseline & $(2,9)$ & $(2,3)$ & $(4,7)$ & $(1,6)$ & $(2,9)$  & $(1,9)$ \\\hline
Ours & $(1,2)$ & $(2,3)$ & $(4,7)$ & $(1,4)$ & $(1,5)$ & $(3,4)$  \\
\bottomrule
\end{tabular}
\vspace{-20pt}
\end{table}

Each agent in each algorithm initializes its Q-values randomly in the interval $[0,50]$. The evolution of their Q-values for state (task) $1$ and actions (robot-pairs) $(1,2)$ and $(1,3)$ is shown in~\Cref{fig:simulation}. Black solid lines correspond to the optimal Q-values computed by the Oracle, solid colored lines to the Baseline, and dashed lines to our method. The results show that under our method, all agents converge almost surely to the true optimal Q-values, whereas the Baseline fails to do so. Consequently, as highlighted in gray in~\Cref{tab:optimal_policy_transposed}, all agents under our method identify the true optimal policies for all states (given by the Oracle), whereas agents under the Baseline fail to do so for states $x=1,4,5,6$.

\section{Conclusion}

\label{sec:conc}
We proposed a Fully Resilient QD-learning (FRQD-learning) algorithm that achieves almost sure convergence to the optimal value functions in distributed multi-agent Q-learning under Byzantine edge attacks. By leveraging a redundancy-based filtering mechanism and the novel topological notion of $(r,r')$-redundancy, our method ensures the induced communication structure remains undirected, thereby recovering the exact convergence properties of the original QD-learning algorithm. For future work, we aim to extend our framework to actor-critic architectures.






\bibliographystyle{IEEEtran}
\bibliography{references_ll}

\end{document}